\documentclass[conference]{IEEEtran}
\usepackage{amsmath, amsthm, amssymb, amsfonts}
\usepackage{bm, bbm, mathrsfs}
\usepackage{dsfont}
\usepackage{graphicx}
\usepackage{algorithm}
\usepackage{algpseudocode}
\usepackage{array}
\usepackage{textcomp}
\usepackage{stfloats}
\usepackage{url}
\usepackage{verbatim}
\usepackage{cite}
\usepackage{xcolor}
\usepackage{tabularx}
\newtheorem{theorem}{Theorem}
\newtheorem{assumption}{Assumption} 
\usepackage{multirow} 
\usepackage{subcaption}
\allowdisplaybreaks[4]

\begin{document}

\title{Multi-Server FL with Overlapping Clients: A Latency-Aware Relay Framework}

\author{
Yun Ji$^{1}$, Zeyu Chen$^{1}$, Xiaoxiong Zhong$^{2}$, Yanan Ma$^{3}$, Sheng Zhang$^{1,*}$, and Yuguang Fang$^{3}$\\[4pt]
$^{1}$Tsinghua Shenzhen International Graduate School, Tsinghua University, Shenzhen 518055, China\\
$^{2}$Department of New Networks, Peng Cheng Laboratory, Shenzhen 518000, China\\
$^{3}$Department of Computer Science, City University of Hong Kong, Kowloon, Hong Kong\\[4pt]
$^{*}$\textit{Corresponding author: Sheng Zhang, email: zhangsh@sz.tsinghua.edu.cn}
}

\maketitle

\begin{abstract}
Multi-server Federated Learning (FL) has emerged as a promising solution to mitigate communication bottlenecks of single-server FL. In a typical multi-server FL architecture, the regions covered by different edge servers (ESs) may overlap. Under this architecture, clients located in the overlapping areas can access edge models from multiple ESs. Building on this observation, we propose a cloud-free multi-server FL framework that leverages Overlapping Clients (OCs) as relays for inter-server model exchange while uploading the local updated model to ESs. This enables ES models to be relayed across multiple hops through neighboring ESs by OCs without introducing new communication links. We derive a new convergence upper bound for non-convex objectives under non-IID data and an arbitrary number of cells, which explicitly quantifies the impact of inter-server propagation depth on convergence error. Guided by this theoretical result, we formulate an optimization problem that aims to maximize dissemination range of each ES model among all ESs within a limited latency. To solve this problem, we develop a conflict-graph-based local search algorithm optimizing the routing strategy and scheduling the transmission times of individual ESs to its neighboring ESs. This enables ES models to be relayed across multiple hops through neighboring ESs by OCs, achieving the widest possible transmission coverage for each model without introducing new communication links. Extensive experimental results show remarkable performance gains of our scheme compared to existing methods.
\end{abstract}


\section{Introduction}
\section{Introduction}
\IEEEPARstart{F}{ederated} Learning (FL) has gained considerable attention due to its ability to enable collaborative model training while preserving data privacy \cite{a3}. A typical FL system operates in rounds, where clients perform local training based on their own datasets and then
upload model updates to a cloud server for aggregation and broadcasting in the next round of training. However, the limitation of this single-server architecture is the long-range communication latency between clients and the cloud server (CS), especially in large-scale FL systems. To address this issue, multi-server FL architectures have been proposed~\cite{a7,a31}. A prominent example is hierarchical federated learning (HFL)~\cite{a7}, where ESs are responsible for aggregating local model updates from their associated clients and subsequently uploading their aggregated models to the CS. However, frequent communications between ESs and the CS can still lead to high latency and resource consumption.

In \cite{a17} and \cite{a40}, a new FL architecture utilizing multiple servers is studied, which exploits the realistic deployment of 5G-and-beyond networks where a client can be located in the overlapping coverage areas of multiple servers, a situation that is particularly common in dense or mesh-like network deployments. The key idea is that clients download multiple models from all the edge servers they can access and train their local models based on the average of these models. Aggregating multiple edge models mainly enhances each OC’s own generalization. Once uploaded to the ES, its contribution to cell-level generalization is limited by the small data size of each client, thus, the global model performance decreases further as the number of OCs drops. To address these challenges, a few studies have begun exploring OC-assisted cross-cell model sharing~\cite{a18}. In \cite{a18}, OCs buffer models from multiple ESs and upload them together after completing local training in the next round, which increases shared information but causes stale updates and higher storage overhead. Building on these insights, we have proposed FedOC~\cite{a51}, where OCs serve as relay nodes that forward edge models between adjacent cells upon reception. Compared with \cite{a17} and \cite{a40}, every client can train on aggregated models from multiple edges in each round, which further enhances generalization and makes the model performance independent of the number of OCs. However, the theoretical results in~\cite{a51} indicate that, as the number of cells increases, relying solely on a limited number of neighboring ES models makes it difficult for each ES to achieve fast synchronization across the network. In this paper, we propose a new approach. Instead of immediately taking actions, witnin a limited time budget, each ES introduces an appropriate waiting period to aggregate its own model with models received from neighboring ESs before forwarding the aggregated result via OCs. Meanwhile, since OCs, as clients, are required to upload their locally trained models to ESs, we aggregate the relayed ES model with the locally uploaded model and forward the aggregated model to neighboring ESs in OCs. In this way, wider model dissemination among ESs can be achieved without introducing any additional communication links or overhead. Motivated by this idea, we derive an upper bound on the optimality gap of the loss function between the proposed distributed aggregation scheme and the optimal solution. Our theoretical analysis shows that the global model becomes closer to the optimal solution as each ES receives more models in the current round. However, waiting for more ES models to be relayed through particular OCs introduces additional latency at each ES, thereby leading to a fundamental tradeoff between model performance and communication latency. Within the framework of leveraging OCs for inter-ES model relaying, we dynamically optimize the model transmission timing of each ES to maximize the number of ES models received by each ES within a finite time period. Without introducing any additional communication links or requiring synchronization with a cloud center, this design promotes distributed model fusion.

In summary, we highlight our main contributions as follows:

\begin{itemize}
    
    \item We propose a distributed aggregated multi-server FL framework, where we exploit OCs as relays for model exchange between neighboring cells and facilitate cross-server collaboration. To reduce relay communication overhead, we assign a single relay OC per overlapping region and merge model relaying with local update uploading into a single transmission. We derive the upper bound of the gap between the expected and optimal global losses with respect to the propagation range of each ES model among all ESs. Based on this, we optimize the model relay strategy and transmission scheduling to neighboring ESs for each ES model, so as to maximize the propagation range of each ES model within a finite time without introducing new communication links.

    \item To solve the optimization problem, we design a conflict-graph-based local search algorithm that jointly optimizes transmission decisions and durations among ESs. By constructing a conflict graph to model transmission interference and applying a greedy–local search strategy, the algorithm maximizes the transmission depth of each ES model among all ESs in a cloud-free manner, thereby improving model propagation efficiency.

    \item We conduct extensive experiments on multiple datasets under varying numbers of cells and heterogeneous data distribution. The results demonstrate that our algorithm consistently outperforms the compared benchmarks in terms of both convergence and accuracy.
\end{itemize}

\section{System Model}
\subsection{Network Architecture}
The system consists of a set of \(L\) ESs indexed by \(\mathcal{L} = \{1,2,\dots, L\}\), where each ES serves a group of clients within its coverage area, and some clients fall within the overlapping regions of multiple ESs. We refer to the local coverage area of each ES as a \textit{cell}. Clients are categorized into two types: local clients (LCs), which are located in non-overlapping regions and communicate with only one ES, and overlapping clients (OCs), which are located in overlapping regions and can communicate with multiple ESs.

\begin{figure}[t]
  \centering
  \includegraphics[width=1\linewidth]{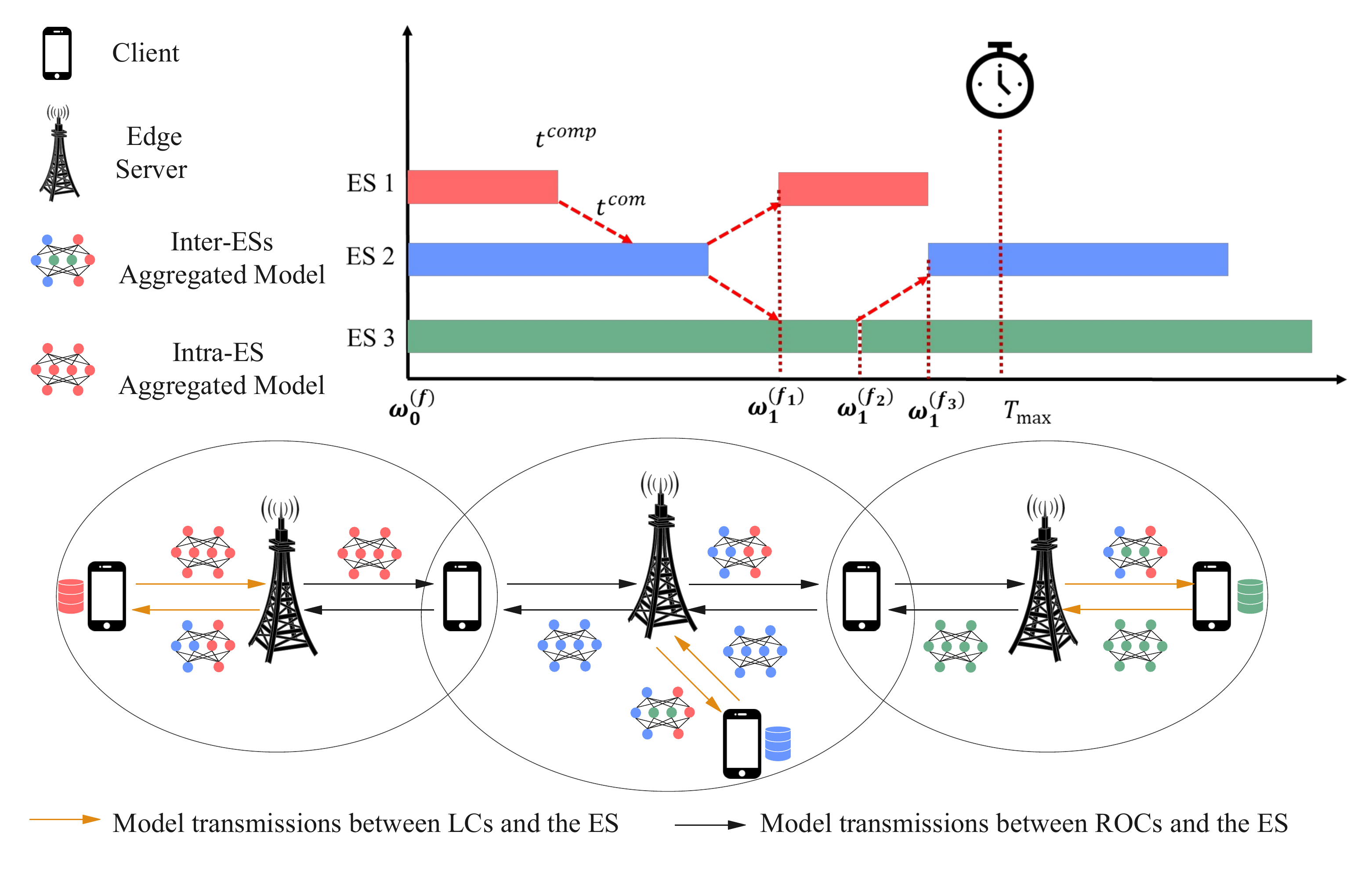}
  \caption{Illustration of our proposed algorithm.}
  \label{fig:system_model}
\vspace{-0.5cm}
\end{figure}

As shown in Fig.~\ref{fig:system_model}, we model overlaps using a chain topology, where each ES \(l\) overlaps with its immediate neighbors, ES \(l-1\) and \(l+1\), except for boundary ESs (ES 1 and ES \(L\)), which overlap with only one neighbor. The chain topology is adopted for two main reasons. First, it simplifies analysis while preserving the key challenges posed by overlapping coverage and multi-server coordination. Second, it can be readily extended to more complex multi-server topologies, e.g., dense mesh deployments where one ES overlaps with three or more neighboring ESs. When a single region is simultaneously overlapped by more than three cells, severe communication interference arises, and the available bandwidth for each cell becomes further fragmented, resulting in substantial communication delays. By limiting each overlapped region to at most two logically participating ESs per communication round, a chain-like logical topology transforms complex mesh structures into simple linear subgraphs, enabling efficient operation without requiring additional complex interference management or communication coordination mechanisms. A single client in each overlap region, denoted \(b_{l,l+1}\), serves as the Relay Overlapping Client (ROC) to forward models between ESs, while the rest are Normal OCs (NOCs). The commonly used notations are summarized in Table~\ref{tab:table0}.


\begin{table}[t]
\centering
\caption{Frequently Used Nomenclature and Notations}  
\label{tab:table0}
\begin{tabularx}{\linewidth}{l|X}
\hline
\textbf{Notation} & \textbf{Description} \\
\hline
\(n^{(k)}\) & The number of data samples of client \(k\) \\
\hline
\( b_{l,l+1} \) & The relay overlapping client (ROC) in the overlapped area between cell \(l\) and \(l+1\) \\
\hline
\(R\) & The total number of rounds \\
\hline
\(\mathcal{S}_{l}\) & The set of clients transmitting local models to ES \(l\) \\
\hline
\(\tilde{N}_{r}^{(f_l)}\) & The total number of data volume clients in \(\mathcal{S}_{l}\) \\
\hline
\(N\) & The total number of data samples of all clients \\
\hline
\(\bm{\tilde{w}}_{r,E}^{(f_l)}\) & The initial ES model of cell \(l\) is obtained by aggregating the models of clients in \(\mathcal{S}_l\). \\
\hline
\( \bm{w}_{r,E}^{(l+1, l)} \) & The model transmitted from ES \(l+1\) to ES \(l\) by ROC \(b_{l+1,l}\)\\
\hline
\( \bm{w}_{r+1}^{(f_l)} \) & The ES model of cell \( l \) for broadcasting to clients at the start of round \( r+1 \) \\
\hline
\end{tabularx}
\vspace{-0.5cm}
\end{table}

\subsection{Algorithm Description}

Now we describe our FL algorithm tailored to the above setup. As shown in Fig.~\ref{fig:system_model}, the training process can be summarized as following stages, i.e., 1) local model update and upload. 2) intra-ES model aggregation. 3) inter-cell model exchange through ROC relays. 4) ES model aggregation and broadcasting.

At the beginning of each round, clients download the edge models from ESs and update their local models, with operations varying depending on client roles: 

For each LC \(k\), local training is performed for \(E\) iterations using stochastic gradient descent (SGD) on the received edge model. The update in the \(e\)-th iteration is computed as

\begin{align} 
\label{equ1}
\bm{w}_{r,e+1}^{(k)} = \bm{w}_{r,e}^{(k)} - \eta_{r,e} \nabla \ell_k (\bm{w}_{r,e}^{(k)}),
\end{align}
where \(e = 0, 1, \ldots, E-1\), and \( \nabla \ell_k (\bm{w}_{r,e}^{(k)})\) is the gradient of the local loss function. After local updating, each LC \(k\) transmits the local model \(\bm{w}_{r,E}^{(k)}\) to corresponding ES. 

For OCs, an overlapping client \(k \in \mathcal{V}_{l,l+1}\) could receive two models: \(\bm{w}_r^{(f_l)}\) and \(\bm{w}_r^{(f_{l+1})}\). To reduce the wait latency, we assign each OC to the nearest ES as its participating server for training. Once it receives the model disseminated by that ES, it performs local training as the same as LCs based on the received model. Finally, each NOC transmits \(\bm{w}_{r,E}^{(k)}\) to selected ES \(i\), ROCs keep their updated models locally and upload them during the subsequent phase. Each ES \(l\) aggregates the received local models as 
\begin{align} 
\label{equ2}
\bm{\tilde{w}}_{r,E}^{(f_l)} = \frac{\sum_{k \in \mathcal{S}_r^{(l)}} n^{(k)} \boldsymbol{w}_{r,E}^{(k)}}{\tilde{N}_{r}^{(f_l)}},
\end{align}
where \(\tilde{N}_{r}^{(f_l)} = \sum_{k \in \mathcal{S}_r^{(l)}} n^{(k)}\), and \(\mathcal{S}_r^{(l)}\) consists of LCs and the NOCs that upload local models to ES \(l\). 

In the stage of inter-cell model exchange, since no additional communication links are introduced, each ES \(l\) can communicate with its ROCs at most once in each direction to relay the model to other ESs. Each ES first decides whether to transmit its model to neighboring ESs, and then determines whether to wait for the models received from its neighbors to be aggregated before forwarding them. In this way, the transmission timing to different neighbors is optimized to achieve the widest possible model propagation within a finite time horizon. The above routing strategy will be elaborated in the section IV. The relay rule of each ROC is: once ROC \( b_{l,l+1} \) receives an edge model from either ES \( l \) or ES \( l+1 \), it will first aggregate this model with its own local model \(\bm{w}_{r,E}^{(b_{l,l+1})}\) trained during the model update stage, and then forward it to the other ES, and we denote the model that ROC \( b_{l, l+1} \) transmits to ES \( l \) as \(\bm{w}_{r,E}^{(l+1, l)}\)


For each ES \(l\), after receiving the models and corresponding data volume from ROCs \(b_{l,l+1} \text{ and } b_{l-1,l}\), ES \(l\) will aggregate these models with its intra-cell model to update edge model \(\bm{w}_{r+1}^{(f_l)}\) for broadcast in \(r+1\)-th iteration, which can be written as 

\begin{align} 
\label{equ3}
\bm {w}_{r+1}^{(f_l)}
= \frac{N_{r}^{(l+1,l)} \bm{w}_{r,E}^{(l+1, l)} + N_{r}^{(l-1,l)} \bm{w}_{r,E}^{(l-1, l)} + \tilde{N}_{r}^{(f_l)} \bm{\tilde{w}}_{r,E}^{(f_l)}}{N_{r}^{(l+1,l)} + N_{r}^{(l-1,l)} + \tilde{N}_{r}^{(f_l)}}.
\end{align}
For boundary cells, the missing neighbor’s model is taken as zero. so \eqref{equ3} holds for all \(l\in\mathcal{L}\).

Noted that \(\bm{w}_{r,E}^{(l-1,l)}\) may contain not only the locally aggregated model of cell \(l-1\) but also models relayed from other cells through multi-hop transmissions. For instance, the model of cell \(i\) can be forwarded along the path \(\mathcal{P}_{i \rightarrow (l-1)} = \{\text{ES } i, \text{ROC } b_{i,i+1}, \text{ES } (i+1), \dots, \text{ROC } b_{l-2,l-1}, \text{ES } (l-1)\}\), where each relay node aggregates the received model with its own before forwarding, provided that the timing constraint is satisfied. We denote \(\bm{p}_{r}^{(f_l)} = [p_{r}^{(1,l)}, \ldots, p_{r}^{(L,l)}]\), where \(p_{r}^{(i,l)} \in \{0,1\}\) indicates whether cell \(i\)’s model participates in ES \(l\)’s aggregation. Based on~\eqref{equ3}, we have:

\begin{align} 
\label{equ4}
\bm{{w}}_{r+1}^{(f_l)} 
&= \frac{1}{N_{r,E}^{(f_l)}} \bigg[\sum_{i=1}^{l-1} p_{r}^{(i,l)}\left( \tilde{N}_r^{(f_i)} \bm{\tilde{w}}_{r,E}^{(f_i)} + n^{(b_{i,i+1})} \bm{w}_{r,E}^{(b_{i,i+1})}\right) \notag\\
&\quad + p_{r}^{(l,l)} \tilde{N}_{r}^{(f_l)} \bm{\tilde{w}}_{r,E}^{(f_l)} \notag\\
&\quad + \sum_{i=l+1}^{L} p_{r}^{(i,l)}\left( \tilde{N}_r^{(f_i)} \bm{\tilde{w}}_{r,E}^{(f_i)}+ n^{(b_{i-1,i})} \bm{w}_{r,E}^{(b_{i-1,i})} \right) \bigg] \notag\\
&= \sum_{i=1}^{L}\frac{p_{r}^{(i,l)}\hat{N}_{r}^{(f_i)} \bm{\hat{w}}_{r,E}^{(f_i)}}{\sum_{i=1}^{L} p_{r}^{(i,l)}\hat{N}_{r}^{(f_i)}},
\end{align}
where \(\bm{\hat{w}}_{r,E}^{(f_i)}\) is denoted as

\begin{align}
\label{equ5}
\bm{\hat{w}}_{r,E}^{(f_i)} = \sum_{k\in \mathcal{\hat{K}}_{r}^{(f_i)}}\frac{n^{(k)}}{\hat{N}_{r}^{(f_i)}} \bm{w}_{r,E}^{(k)},
\end{align}
where \(\hat{N}_{r}^{(f_i)} = \sum_{k\in \mathcal{\hat{K}}_{r}^{(f_i)}} n^{(k)}\), and the client set \(\hat{\mathcal{K}}_{r}^{(f_i)}\) satisfies:

\begin{align}
\label{equ6}
\mathcal{\hat{K}}_{r}^{(f_i)} = 
\begin{cases} 
\mathcal{S}_{i} \cup b_{i,i+1}, & 1 \leq i < l \\
\mathcal{S}_{i}, & i = l \\
\mathcal{S}_{i} \cup b_{i-1,i}, & l+1 \leq i \leq L. 
\end{cases}
\end{align}

\subsection{Latency Analysis}
Several key event timings are critical to our algorithm:
\begin{itemize}
    \item \textbf{Broadcast time (\(t_{l}^{\textnormal{cast}}\))}: Time for ES \(l\) to broadcast its model to its covered clients.
    \item \textbf{Cell update time (\(t_{l}^{\textnormal{comp}}\))}: Time for all clients in a cell to complete local computations and uploading model updates for aggregation.
    \item \textbf{Relay start time (\(t_{l,j}^{\textnormal{start}}\))}: Time when ES \(l\) begins  transmitting its local aggregated ES model to a neighboring cell \(j \in \{l-1, l+1\}\) through the ROC \(b_{l,j}\).
    \item \textbf{Communication transmission time (\(t_{l,j}^{\textnormal{com}}\))}: Time taken for ES \(l\) to transmit its model to a neighboring cell \(j\) through the ROC \(b_{l,j}\).
    \item \textbf{Model aggregation time (\(t_{l}^{\textnormal{agg}}\))}: Time when ES \(l\) final aggregates all the received ES models.
\end{itemize}

Assume a total spectrum budget of \(B\) for model transmission. Adjacent cells operate on disjoint sub-bands, each allocated \(B/2\) bandwidth. Each ES broadcasts the model to its clients and uniformly allocates the available bandwidth among them for local model uploads. After aggregation, the ES reclaims the bandwidth and forwards the model to neighboring ESs via ROCs. The communication latency can be given by

\begin{align}
\label{equ7}
t_{l,l+1}^{\text{com}} = \frac{M}{\frac{B}{4}\left[ \log\left(1+\frac{4\delta_{l,l+1}P}{BN_0}\right) + \log\left(1+\frac{4\delta_{l,l+1}p}{BN_0}\right) \right]},
\end{align}
where \(M\) is the model size, and \(\delta_{l,l+1}\) represents the channel gain from ES \(l\) to \(l+1\), with \(P\) and \(p\) being the transmission powers of the ES and client, respectively. \(N_0\) is the noise power spectral density. The relay start time is constrained by

\begin{align} 
\label{equ8}
t_{l,j}^{\text{start}} \geq t_{l}^{\textnormal{cast}} + t_{l}^{\textnormal{comp}}, \quad \forall j \in \{\,l-1,\,l+1\,\}.
\end{align}



The aggregation time \(t_{l}^{\textnormal{agg}}\) is determined by:

\begin{align} 
\label{equ9}
t_{l}^{\textnormal{agg}} = \max \left\{ p^{(j,l)} \big( t_{j,l}^{\textnormal{start}} + t_{j,l}^{\textnormal{com}} \big) \;\middle|\; j \in \{l-1, l, l+1\} \right\}.
\end{align}

\section{THEORETICAL RESULTS}
In this paper, we consider a \(C\)-class classification problem over a compact feature space \(\mathcal{X}\) with label space \(\mathcal{Y}=\{1,2,\ldots,C\}\). Following \cite{a34}, we define the population loss under the cross-entropy as \(\ell(\bm{w})=\sum_{i=1}^C P_{y=i}\, \mathbb{E}_{\bm{x}\mid y=i}\big[\log f_i(\bm{x},\bm{w})\big]\), where \(P_{y=i}\) denotes the probability of class \(i\), and \(f_i\) is the predicted probability of class \(i\) parameterized by the model weights \cite{a34}.

Now we introduce the widely used assumptions in FL:
\begin{assumption}
\label{assumption1} \(\nabla_{\boldsymbol{w}} \mathbb{E}_{\boldsymbol{x} \mid y=i} \log f_i(\boldsymbol{x}, \boldsymbol{w})\) is \(\lambda_{\boldsymbol{x} \mid y=i}\)-Lipschitz for each class \(i \in [C]\), i.e., for all \(\bm{v}\) and \(\bm{w}\), \(\|\nabla_{\boldsymbol{w}} \mathbb{E}_{\boldsymbol{x} \mid y=i} \log f_i(\boldsymbol{x}, \boldsymbol{w}) - \nabla_{\boldsymbol{w}} \mathbb{E}_{\boldsymbol{x} \mid y=i} \log f_i(\boldsymbol{x}, \boldsymbol{v}) \| \leq \lambda_{\boldsymbol{x} \mid y=i}\|\bm{w}-\bm{v}\|\).
\end{assumption}

\begin{theorem}
\label{theorem1}
Let Assumption~\ref{assumption1} holds, assume \(\sum_{i=1}^{C}P_{y=i}\lambda_{x|y=i} = \lambda\), and set \(\bm{w}^{*}\) to the optimal global model. If the learning rate satisfies \(\eta_{r,e} = \frac{1}{(r+1)(E-1)}~\forall~r,~e\), then we have
\begin{align} 
\label{equ10}
& \ell(\bm{w}_{R}^{(f_l)}) - \ell(\bm{w^{*}}) \notag\\
\leq & \frac{\lambda}{2} \Bigg[ \sum_{j=1}^{L} 
D_{R-1}^{(j)}  \left\|\boldsymbol{w}_{R-1}^{(f_j)}-\boldsymbol{w}_{R-1}^{(c)}\right\| + \underbrace{\frac{\beta_{R-1}}{R}}_{\epsilon^{\text{intra}}} \notag\\
&+ \underbrace{\frac{\sum_{j=1}^{3} \hat{N}_{R-1}^{(f_j)} H^{(j)}\left( \sum_{i=1}^{C} \left\| P_{y=i}^{(c_j)} - P_{y=i}^{(c)} \right\|\right) }{NR(E-1)} + F_{R-1}^{(l)}}_{\epsilon^{\text{inter}}} \Bigg],
\end{align} 
where \( 
F_{R-1}^{(l)} = 
\sum_{j=1}^{L}\left|\frac{p_{R-1}^{(j,q)}\hat{N}_{R-1}^{(f_j)}}
{\sum_{j=1}^{L} p_{R-1}^{(j,q)}\hat{N}_{R-1}^{(f_j)}} 
- 
\frac{\hat{N}_{R-1}^{(f_j)}}
{\sum_{j=1}^{L} \hat{N}_{R-1}^{(f_j)}}\right|
\!\left\|\bm{\hat{w}}_{R-1,E}^{(f_j)}\right\|,\)
and \(
\beta_{R-1} = 
\frac{E-1}{N}\sum_{j=1}^{L}
\sum_{k\in\hat{\mathcal{K}}_{r}^{(f_j)}}n^{(k)}
\Bigl(\sum_{i=1}^C|P_{y=i}^{(k)}-P_{y=i}^{(c_j)}|\Bigr)
\!\left(\prod_{e=1}^{E-1}a_{R-1,e}^{(k)}\right)
g_{\max}\!\bigl(\boldsymbol{w}_{R-1}^{(c_j)}\bigr).
\)
\end{theorem}
\begin{proof}
Please refer to the Appendix.
\end{proof}
\(\epsilon^{\text{intra}}\) quantifies the intra-cell heterogeneity, mainly reflected by the divergence between each client’s local data distribution and that of the entire cell, i.e., \(\sum_{i=1}^{C}\left| P_{y=i}^{(k)} - P_{y=i}^{(c_j)} \right|\) for \(k \in \hat{\mathcal{K}}_{r}^{(f_j)}\), \(\forall j\). This term vanishes as \(R \to \infty\). \(\epsilon^{\text{inter}}\) captures the inter-cell heterogeneity across different cells and consists of two parts. The first part corresponds to the divergence between each cell’s distribution and that of the entire system, i.e., \(\sum_{i=1}^{C} \left\| P_{y=i}^{(c_j)} - P_{y=i}^{(c)} \right\|\), which converges to zero as \(R \to \infty\). The second part, \(F_{R-1}^{(l)}\), arises from the model aggregation procedure in our algorithm. As the number of external cell models incorporated into \(\bm{w}_{r}^{(f_l)}\) increases, i.e., more elements in \(\bm{p}_{r}^{(f_l)} = [p_{r}^{(1,l)}, \ldots, p_{r}^{(L,l)}]\) equals to 1, \(F_{r}^{(l)}\) approaches zero. In the extreme case where every ES model \(\bm{w}_{r}^{(f_l)}\) aggregates all ES models at each round, i.e., \(\bm{p}_{r}^{(f_l)},~\forall~l\) is an all-one vector, the process degenerates into centralized FL with full cloud aggregation, and thus \(F_{r}^{(l)} = 0\). To sum up, the loss error is significantly influenced by the probability of each cell participating in the aggregation of other cell models in each round.
\section{Problem Formulation and Solution}
Based on \textbf{Theorem~1}, we need to minimize \(F_{r}^{(l)}\) for all \(l\), which is equivalent to maximizing the propagation range of each ES model within a limited time. This is achieved by optimizing the routing decisions and the relay start times of each ES, enabling broader model sharing among cells.

The transmission path from cell \(j\) to \(l\) is defined as 
\(\text{path} = \{\text{ES } j, \text{ES } j+1, \dots, \text{ES } l\}\).
Successful transmission requires: 
(1) a valid link between adjacent nodes, i.e., \(p^{(q,q+1)}=1\),
(2) each model arriving before the next forwarding step, i.e., 
\(t_{q,q+1}^{\text{start}} + t_{q,q+1}^{\text{com}} \leq t_{q+1,q+2}^{\text{start}}\).

We define the indicator variable:
\begin{align} 
\label{equ11}
s^{(q,q+1)} = \mathds{1}\Big\{ t_{q,q+1}^{\text{start}} + t_{q,q+1}^{\text{com}} \leq t_{q+1,q+2}^{\text{start}} \Big\}.
\end{align}

Then, we have
\begin{align} 
\label{equ12}
p^{(j,l)} = p^{(l-1,l)} \prod_{q=j}^{l-2} \big(s^{(q,q+1)} p^{(q,q+1)}\big),~j \leq l-1,
\end{align}
\begin{align} 
\label{equ13}
p^{(j,l)} = p^{(l+1,l)} \prod_{q=j}^{l+2} \big(s^{(q,q-1)} p^{(q,q-1)}\big),~l+1 \leq j \leq L.
\end{align}

We denote the optimization problem as
\begin{align}
\label{equ14}
\textbf{P1:}\quad 
& \min_{\bm{t}^{\text{start}}, \bm{p}} 
&& \sum_{l=1}^{L} \sum_{j=1}^{L} 
\left| 
\frac{p^{(j,l)} \hat{N}^{(f_j)}}{\sum_{j=1}^L p^{(j,l)} \hat{N}^{(f_j)}} 
- 
\frac{\hat{N}^{(f_j)}}{\sum_{j=1}^L \hat{N}^{(f_j)}} 
\right| \nonumber \\
& \text{s.t.} 
&& \eqref{equ8},~\eqref{equ9},~\eqref{equ11}-\eqref{equ13}, \nonumber \\
& 
&& t_{l}^{\text{agg}} \le T_{\text{max}}, \quad \forall l \in \mathcal{L}.
\end{align}

where \(\bm{t}^{\text{start}}=\{t_{l,j}^{\textnormal{start}}\},~\forall~l,j\), \(\bm{p}=\{\bm{p}^{(f_l)}\},~\forall~l\) and \(T_{\text{max}}\) denotes the maximum duration of each round.

Since the transmission processes in opposite directions (toward left and right ESs) are independent in the constraints, the variables for each direction can be optimized separately. After removing the absolute value, \textbf{P1} can be reformulated to 

\begin{align}
\label{equ15}
\textbf{P2:}\quad 
& \max_{\bm{t}_{\text{right}}^{\text{start}},\, \bm{p}_{\text{right}}}
&& \sum_{l=1}^{L} \sum_{j=1}^{l}
\left( p^{(j,l)} - \mathds{1}\!\left\{p^{(j,l)} = 0\right\} \right) \hat{N}^{(f_j)} \nonumber \\
& + \max_{\bm{t}_{\text{left}}^{\text{start}},\, \bm{p}_{\text{left}}}
&& \sum_{l=1}^{L} \sum_{j=l+1}^{L}
\left( p^{(j,l)} - \mathds{1}\!\left\{p^{(j,l)} = 0\right\} \right) \hat{N}^{(f_j)} \nonumber \\
& \text{s.t.} 
&& \eqref{equ8},~\eqref{equ9},~\eqref{equ11}-\eqref{equ13}, \nonumber \\
& 
&& t_{l}^{\text{agg}} \le T_{\text{max}}, \quad \forall l \in \mathcal{L}.
\end{align}
where 
\(\bm{t}_{\text{right}}^{\text{start}} = \{t_{l,j}^{\textnormal{start}}\}\), 
\(\bm{p}_{\text{right}} = \{p^{(l,j)}\}\), \(\forall~l \le j\), \(\bm{t}_{\text{left}}^{\text{start}} = \{t_{l,j}^{\textnormal{start}}\}\),  
\(\bm{p}_{\text{left}} = \{p^{(l,j)}\}\), \(\forall~l > j\).
We illustrate the algorithm using the first problem as an example. Firstly \(\{ p^{(j,j+1)} \}\) for \(1 \leq j \leq L-1\) can be determined straightforwardly. If \(t_{j,j+1}^{\text{local}} + t_{j,j+1}^{\text{com}} \leq T_{\text{max}}\), then \(p^{(j,j+1)} = 1\). This ensures that ES \(j+1\) at least can receive the model from ES \(j\). If \(p^{(j,l)}=1\), all intermediate relays from ES \(j\) to ES \(l\) must also be active. Thus, maximizing the objective is equivalent to finding the longest feasible relay paths. To achieve this, each ES \(j\) transmits its cell model to ES \(j+1\) once intra-cell aggregation is completed. 
After receiving the model, ES \(j+1\) forwards the aggregated model at 
\(\max \{t_{j,j+1}^{\text{start}} + t_{j,j+1}^{\text{com}},~t_{j+1}^{\text{cast}} + t_{j+1}^{\text{comp}} \}\), 
ensuring \(s^{(j,j+1)} = 1\) as defined in~\eqref{equ11}. 
This process continues until \(s_{q,q+1} = 0\), where no feasible \(t_{q,q+1}^{\text{start}}\) satisfies 
\(t_{q,q+1}^{\text{start}} + t_{q,q+1}^{\text{com}} \le T_{\text{max}}\). 
The resulting longest path from ES \(j\) is \(\{j, j+1, \ldots, q\}\). 
Accordingly, the optimal transmission path for each ES \(l\) is 
\(\mathcal{P}_{\text{right}}^{(q,l)} = \{q, q+1, \ldots, l\}\), 
where \(q\) is the earliest ES that can successfully transmit its model to ES \(l\). Conflicts arise when two paths share intermediate nodes. We model this as a conflict graph \(\mathcal{G}(\mathcal{V},\mathcal{E})\), where vertices represent paths weighted by the total data volume of cells along this path, which we denote as \(D^{(q,l)}\) for path \(\mathcal{P}^{(q,l)}\), and edges denote conflicts. The problem is then equivalent to selecting a maximum-weight independent set. For small networks, exhaustive search provides the optimum. For larger networks, we propose a Conflict-Graph-Based Local Search Algorithm. The algorithm consists of three steps as follows.

\textbf{Step 1.} Generate an initial independent set via a greedy selection of non-conflicting high-weight vertices. 

\textbf{Step 2.} Sort all paths by their end node numbers in ascending order, and find the sub-optimal and non-conflicted path for each node between the end node of one path and the start node of the next path. Then adding these paths into \(\mathcal{C}  \left( \mathcal{I}^{\text{ini}} \right) \) which we named as the full paths set of \(\mathcal{I}^{\text{ini}}\), and calculate the objective of \textbf{P2}, which we denote as \(U \left( \mathcal{I}^{\text{ini}} \right)\). 

\textbf{Step 3.} Iteratively refine the set by local search update higher \(U \left( \mathcal{I}^{\text{final}} \right)\). 

The detailed algorithm is summarized in \textbf{Algorithm~\ref{alg:conflict}}.

\begin{algorithm}[t]
\caption{Conflict-Graph-Based Local Search}
\label{alg:conflict}
\textbf{Input:} Routing decision \(\bm{p}_{\text{right}}\) (or \(\bm{p}_{\text{left}}\)), 
paths \(\mathcal{P}^{(q,l)}\) with weights \(D^{(q,l)}\). \\
\textbf{Output:} Relay start times \(\bm{t}_{\text{right}}^{\text{start}}\) (or \(\bm{t}_{\text{left}}^{\text{start}}\)).
\begin{algorithmic}[1]
\State Initialize independent set \(\mathcal{I}^{\text{ini}}\) greedily by using \textbf{Step 1}.
\State Compute objective \(U(\mathcal{I}^{\text{ini}})\).
\For {each \(i \in \mathcal{I}^{\text{ini}}\)}
    \State Replace \(i\) by feasible \(j\) to form \(\mathcal{I}^{\text{new}}\), s.t. \(\{j, \mathcal{I} \setminus i\}\) still remain independent.
    \State Calculate \(\mathcal{C}  \left( \mathcal{I}^{\text{new}} \right) \) and \(U \left(\mathcal{I}^{\text{new}}\right)\) by using \textbf{step 2}.
    \If { \(U(\mathcal{I}^{\text{new}}) > U(\mathcal{I}^{\text{final}})\) }
        \State Update \(\mathcal{I}^{\text{final}} = \mathcal{I}^{\text{new}}\).
    \EndIf
\EndFor
\State Calculate \(\mathcal{C}  \left( \mathcal{I}^{\text{final}} \right) \) and \(U\left(\mathcal{I}^{\text{final}}\right)\) by using \textbf{step 2}.
\For {each node in path \(\mathcal{P}^{(q,l)} \in \mathcal{C}(\mathcal{I}^{\text{final}})\)}
    \State \(t_{j,j+1}^{\text{start}} = \max(t_{j-1,j}^{\text{start}}+t_{j-1,j}^{\text{com}},~ t_{j}^{\text{cast}}+t_{j}^{\text{comp}})\).
\EndFor
\State \Return \(\bm{t}_{\text{right}}^{\text{start}}\) (or \(\bm{t}_{\text{left}}^{\text{start}}\)).
\end{algorithmic}
\end{algorithm}

\section{Performance Evaluation}
\subsection{Simulation Setup}
We evaluate a multi-server network with \(L\) ESs and \(K=60\) clients, each ES covering a circular area of radius 600m, where lients are randomly distributed. The wireless channel model includes Rayleigh small-scale fading and large-scale path-loss (\(128.1+37.6\log_{10}(d(\text{km}))\)). Experiments are conducted on MNIST and CIFAR-10 with heterogeneous data distributions: each client has samples from two classes, and each ES is restricted to five classes, creating strong imbalance. The maximum per-round time \(T_{\text{max}}\) is aligned with FedOC, i.e., the maximum time for  for fair comparison. For the MNIST dataset, we adopt a lightweight Convolutional Neural Network (CNN) with 21840 model parameters~\cite{a7}. For CIFAR-10, we employ a deeper six-layer CNN with 1.14 million parameters~\cite{a31}. Key parameters are summarized in Table~\ref{tab1}.



We compare our method with the following baselines:
\begin{itemize}
    \item \textbf{HFL\cite{a7}}: Hierarchical FL without overlapping regions.
    \item \textbf{FedMES\cite{a17}}: OCs train local models based on the aggregated multiple ES models.
    \item \textbf{FL-EOCD~\cite{a18}}: After caching and training on aggregation of received ES models, each OC uploads a model that combines its local update with the cached ES models.
    \item \textbf{FedOC~\cite{a44}}: ESs directly forward aggregated models to neighbors via ROCs after local aggregation, without transmission time optimization.
\end{itemize}
\begin{table}[t]
\caption{Simulation Parameters}
\begin{center}
\small
\begin{tabular}{p{5.4cm}|p{2.7cm}}  
\hline
\textbf{Parameters} &    \textbf{Values} \\
\hline
Number of clients, \(K\) & \(60\)\\
\hline
Number of iterations, \(R\) & \(500\)\\
\hline
Number of local training epochs, \({U}_{k}(r)\) & \(5\)\\
\hline
Client's transmit power, \(p~\mathrm{(W)}\) & \(1\) \\
\hline
ES's transmit power, \(P~\mathrm{(W)}\) & \(5\)\\
\hline
Total channel bandwidth, \(B~\mathrm{(MHz)}\)  & \(50\)\\
\hline
Client's one-epoch update time (seconds) & MNIST: \([0.1, 0.2]\), CIFAR-10: \([1, 2]\) \\
\hline
Noise power spectral density, \({N}_{0}~\mathrm{(dBm/Hz)}\) & \(-174\)\\
\hline
The number of model parameters & MNIST: \(21840\), CIFAR-10: \(1.14\) million \\
\hline
Batch size & \(20\)\\
\hline
Initial learning rate, \(\eta\) & MNIST: \(0.01\), CIFAR-10: \(0.1\)\\
\hline 
Exponential decay factor of \(\eta\) in SGD & MNIST: \(0.995\), CIFAR-10: \(0.992\)\\
\hline
\end{tabular}
\label{tab1}
\end{center}
\end{table}

As shown in Table~\ref{tab:avg_clients}, we compare the average number of clients aggregated per cell between FedOC and our method. Under the same round time \(T_{\text{max}}\), our approach enables each cell to receive significantly more ES models, including those from distant cells,  without extra communication links. This confirms the effectiveness of our algorithm.

\begin{table}[t]
\centering
\caption{The average number of clients aggregated per cell.}
\label{tab:avg_clients}
\begin{tabular}{c|c|c|c}
\hline
\textbf{Dataset} & \textbf{Cells} & \textbf{FedOC} & \textbf{Ours} \\
\hline
\multirow{3}{*}{MNIST} 
& 3 cells & 47.67 & \textbf{53.00} \\
& 5 cells & 31.40 & \textbf{50.00} \\
& 6 cells & 26.76 & \textbf{50.00} \\
\hline
\multirow{3}{*}{CIFAR-10} 
& 3 cells & 46.67 & \textbf{55.33} \\
& 5 cells & 31.40 & \textbf{53.40} \\
& 6 cells & 27.79 & \textbf{34.35} \\
\hline
\end{tabular}
\end{table}
\begin{figure*}[t]
\centering
\begin{subfigure}[t]{0.31\textwidth}
    \centering
    \includegraphics[width=\linewidth]{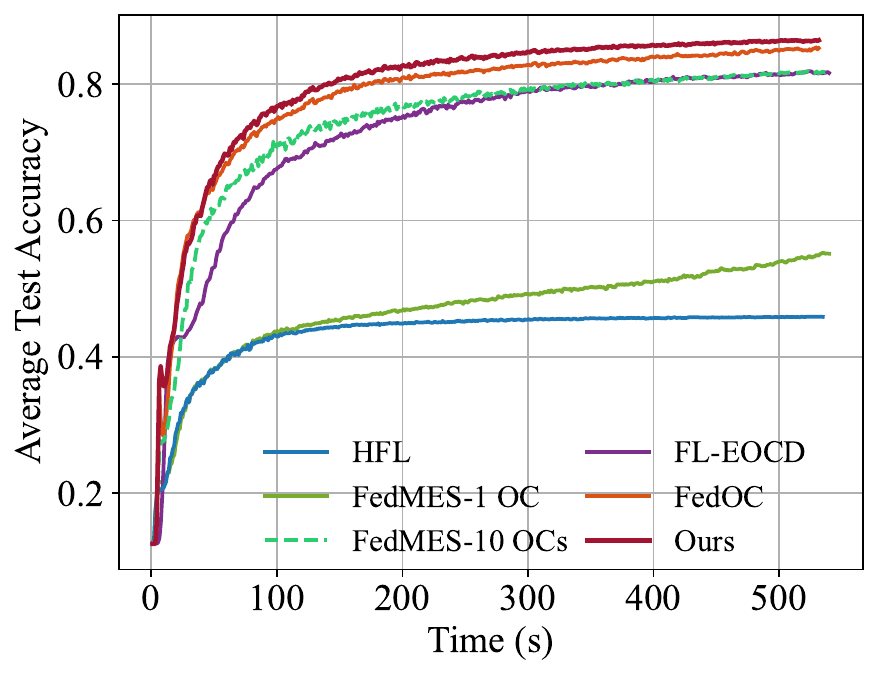}
    \caption{MNIST, 3cells}
    \label{MNIST, 3cells}
\end{subfigure}
\quad
\hfill
\begin{subfigure}[t]{0.31\textwidth}
    \centering
    \includegraphics[width=\linewidth]{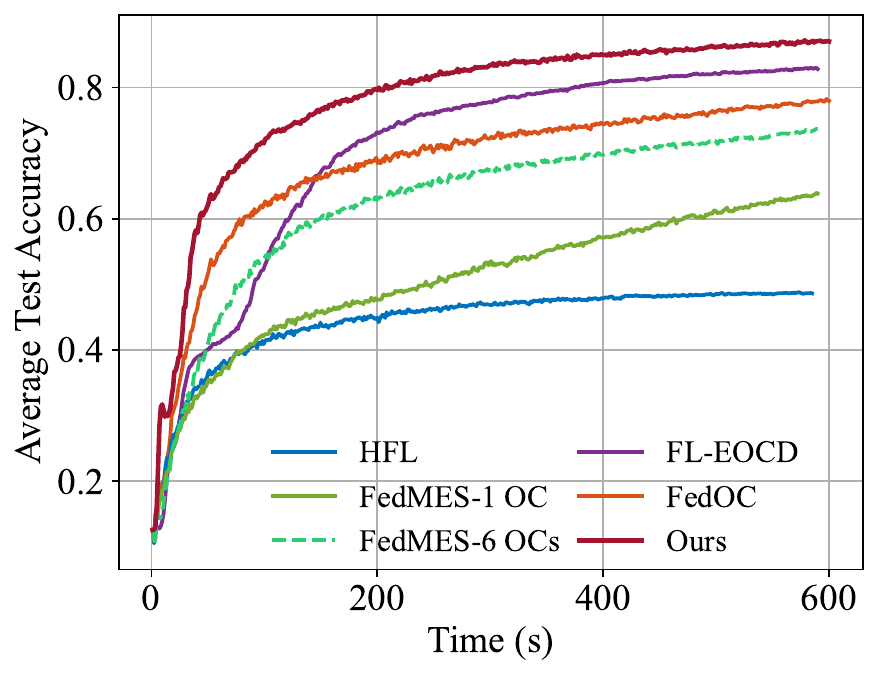}
    \caption{MNIST, 5cells}
    \label{MNIST, 5cells}
\end{subfigure}
\quad
\hfill
\begin{subfigure}[t]{0.31\textwidth}
    \centering
    \includegraphics[width=\linewidth]{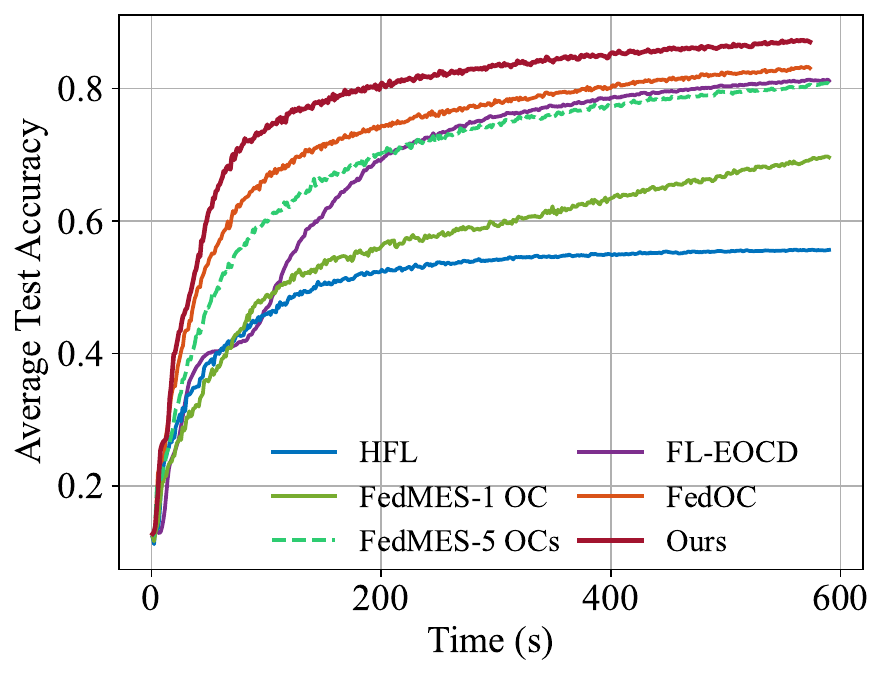}
    \caption{MNIST, 6cells}
    \label{MNIST, 6cells}
\end{subfigure}
\begin{subfigure}[t]{0.31\textwidth}
    \centering
    \includegraphics[width=\linewidth]{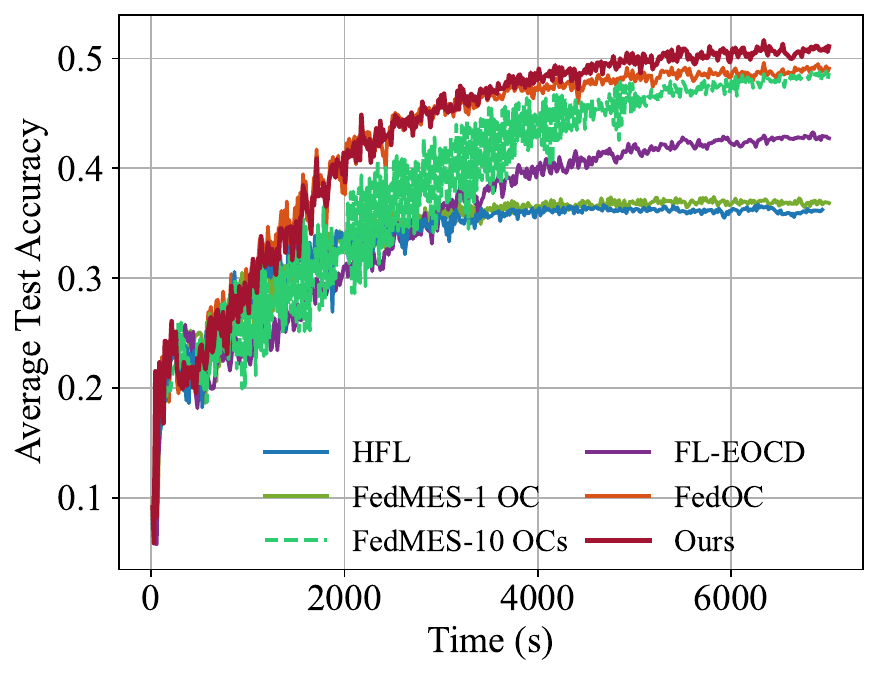}
    \caption{CIFAR-10, 3cells}
    \label{CIFAR-10, 3cells}
\end{subfigure}
\quad
\hfill
\begin{subfigure}[t]{0.31\textwidth}
    \centering
    \includegraphics[width=\linewidth]{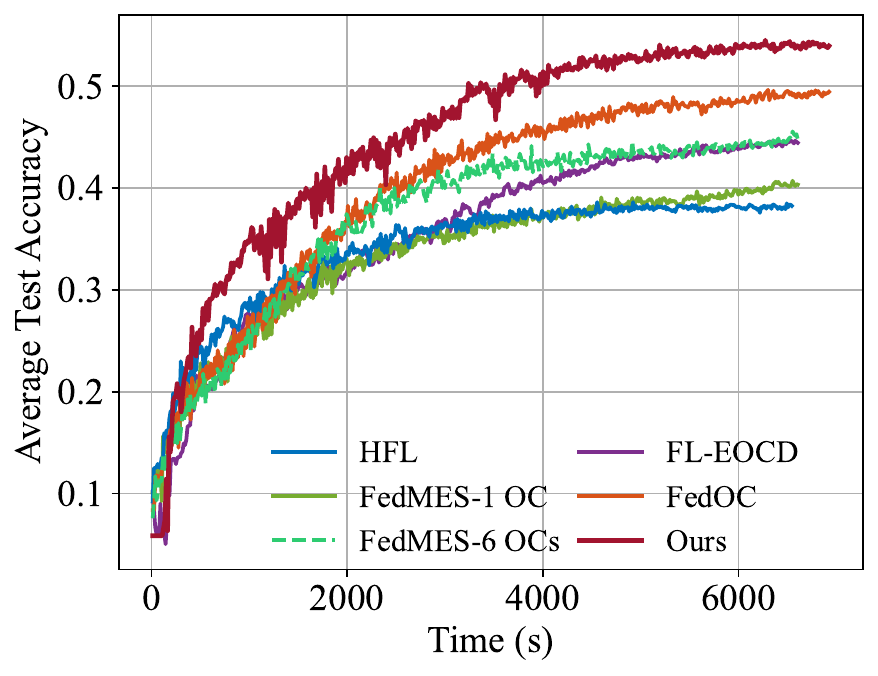}
    \caption{CIFAR-10, 5cells}
    \label{CIFAR-10, 5cells}
\end{subfigure}
\quad
\hfill
\begin{subfigure}[t]{0.31\textwidth}
    \centering
    \includegraphics[width=\linewidth]{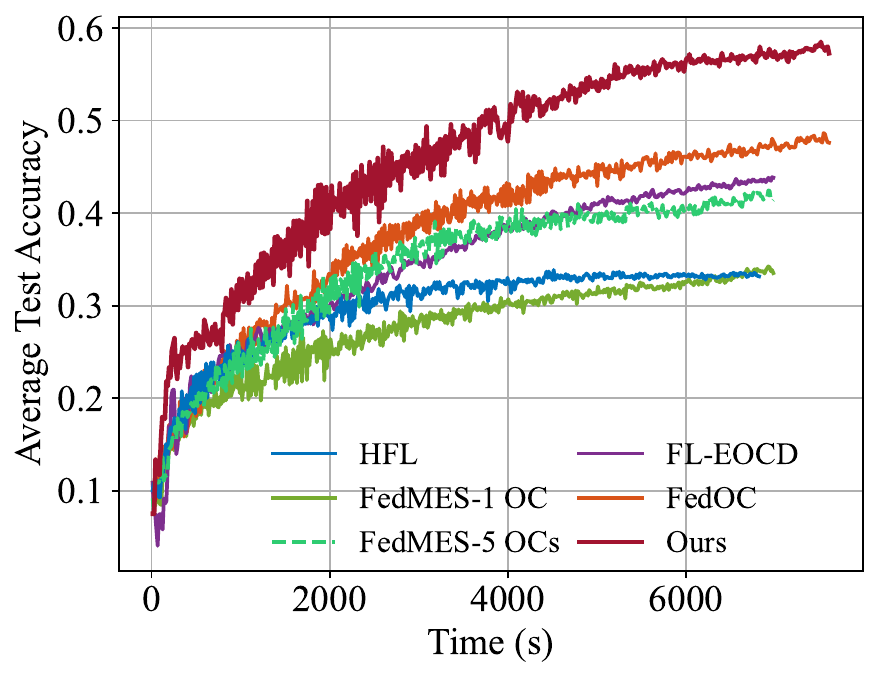}
    \caption{CIFAR-10, 6cells}
    \label{CIFAR-10, 6cells}
\end{subfigure}
\caption{Average test accuracy versus training time}
\label{fig:IID-0-overlap1}
\end{figure*}

Fig.~\ref{fig:IID-0-overlap1} compares the accuracy of different benchmarks across various numbers of cells. HFL and FedMES with one OC only perform intra-cell aggregation in each round, thus their models plateau at a lower accuracies. When the OCs increase to \(\lfloor K / (2L) \rfloor\) clients per region, the accuracy of FedMES increases significantly, benefiting from its reliance on multiple OCs to aggregate ES models, which indirectly enhances the generalization of each ES model. However, relying solely on the limited aggregation within OCs provides only constrained gains in model generalization, especially when the number of cells increases. FL-EOCD fares better: by integrating other cells’ models, albeit with a delay, it achieves higher accuracy than HFL/FedMES in this scenario. In some MNIST cases (e.g., 5 cells) it can even outperform FedOC. However, its advantage diminishes on the more challenging CIFAR-10 task due to model staleness. In contrast, FedOC leverages real-time transmission between neighboring ESs via ROCs, achieving performance second only to our algorithm in most cases. Our algorithm builds upon FedOC by incorporating optimizations for the routing and transmission times of model exchanges, enabling models to be transmitted to more distant ESs without adding extra communication links or time. This optimization accelerates model fusion, leading to faster convergence and higher model accuracy.


\section{Conclusion}
In this paper, we presented a cloud-free multi-server FL framework leveraging overlapping clients as relays to enable multi-hop model dissemination among ESs without introducing additional communication links. We derived a new convergence bound of loss error of our algorithm, which characterizes the relationship between each ES model propagation depth among all ESs and the convergence error. Building upon this, we formulate an optimization problem maximizing the propagation range of edge models under latency constraints by jointly scheduling inter-server route decisions and transmission times. To solve this problem, we proposed a conflict-graph-based local search algorithm that models path dependencies and transmission conflicts as a maximum-weight independent set problem. By combining greedy initialization and iterative local refinement, it efficiently derives near-optimal relay schedules with low computational complexity. Experimental results demonstrated that the proposed framework consistently outperforms existing schemes, achieving faster convergence and higher accuracy. The gains become more pronounced as the number of cells increases, confirming that optimized relay scheduling enables deeper dissemination of models across multiple ESs and thereby accelerates convergence.

{\appendix[Proof of the Theorem 1]
\begin{proof}

Notice that \(\bm{w}_{r,E}^{(l-1,l)}\) in (3)  may contain not only the locally aggregated model of cell \(l-1\) but also models relayed from other cells through multi-hop transmissions. For instance, the model of cell \(i\) can be forwarded along the path \(\mathcal{P}_{i \rightarrow (l-1)} = \{\text{ES } i, \text{ROC } b_{i,i+1}, \text{ES } (i+1), \dots, \text{ROC } b_{l-2,l-1}, \text{ES } (l-1)\}\), where each relay node aggregates the received model with its own model before forwarding, provided that the timing constraint is satisfied. We denote \(\bm{p}_{r}^{(f_l)} = [p_{r}^{(1,l)}, \ldots, p_{r}^{(L,l)}]\), where \(p_{r}^{(i,l)} \in \{0,1\}\) indicates whether cell \(i\)’s model participates in ES \(l\)’s aggregation. Based on~\eqref{equ3}, we have:

\begin{align} 
\label{equ4}
\bm{{w}}_{r+1}^{(f_l)} 
&= \frac{1}{N_{r,E}^{(f_l)}} \bigg[\sum_{i=1}^{l-1} p_{r}^{(i,l)}\left( \tilde{N}_r^{(f_i)} \bm{\tilde{w}}_{r,E}^{(f_i)} + n^{(b_{i,i+1})} \bm{w}_{r,E}^{(b_{i,i+1})}\right) \notag\\
&\quad + p_{r}^{(l,l)} \tilde{N}_{r}^{(f_l)} \bm{\tilde{w}}_{r,E}^{(f_l)} \notag\\
&\quad + \sum_{i=l+1}^{L} p_{r}^{(i,l)}\left( \tilde{N}_r^{(f_i)} \bm{\tilde{w}}_{r,E}^{(f_i)}+ n^{(b_{i-1,i})} \bm{w}_{r,E}^{(b_{i-1,i})} \right) \bigg] \notag\\
&= \sum_{i=1}^{L}\frac{p_{r}^{(i,l)}\hat{N}_{r}^{(f_i)} \bm{\hat{w}}_{r,E}^{(f_i)}}{\sum_{i=1}^{L} p_{r}^{(i,l)}\hat{N}_{r}^{(f_i)}},
\end{align}
where \(\bm{\hat{w}}_{r,E}^{(f_i)}\) is denoted as

\begin{align}
\label{equ5}
\bm{\hat{w}}_{r,E}^{(f_i)} = \sum_{k\in \mathcal{\hat{K}}_{r}^{(f_i)}}\frac{n^{(k)}}{\hat{N}_{r}^{(f_i)}} \bm{w}_{r,E}^{(k)},
\end{align}
where \(\hat{N}_{r}^{(f_i)} = \sum_{k\in \mathcal{\hat{K}}_{r}^{(f_i)}} n^{(k)}\), and the client set \(\hat{\mathcal{K}}_{r}^{(f_i)}\) satisfies:

\begin{align}
\label{equ6}
\mathcal{\hat{K}}_{r}^{(f_i)} = 
\begin{cases} 
\mathcal{S}_{i} \cup b_{i,i+1}, & 1 \leq i < l \\
\mathcal{S}_{i}, & i = l \\
\mathcal{S}_{i} \cup b_{i-1,i}, & l+1 \leq i \leq L. 
\end{cases}
\end{align}

From~\eqref{equ6}, the set \(\mathcal{\hat{K}}_{r}^{(f_i)}\) is not only associated with cell \(i\) but also depends on cell \(l\), which complicates the analysis. Since this difference only arises from the assignment of a single ROC and does not affect the overall client data of each cell, we approximate by assigning each ROC to the left cell in the mathematical formulation. That is,
\begin{align}
\label{equ19}
\mathcal{\hat{K}}_{r}^{(f_i)} = \mathcal{S}_{i} \cup b_{i,i+1}, \quad 1 \leq i < L-1.
\end{align}

Under these definitions, each ROC is mathematically associated with a unique cell. As shown in~\eqref{equ4} the edge model of each ES \(l\) is the weighted sum of all participating ES models.

Since the population loss \(\ell (\bm{w})\) satisfies \(\ell(\boldsymbol{w})=\mathbb{E}_{\boldsymbol{x}, y \sim P}\left[\sum_{i=1}^C \mathbbm{1}_{y=i} \log f_i(\boldsymbol{x}, \boldsymbol{w})\right]=\sum_{i=1}^C P_{y=i} \mathbb{E}_{\boldsymbol{x} \mid y=i}\left[\log f_i(\boldsymbol{x}, \boldsymbol{w})\right]\), the local update in~\eqref{equ1} can be written as

\begin{align} 
\label{equ20}
\bm{w}_{r,e+1}^{(k)} = \bm{w}_{r,e}^{(k)} - \eta_{r,e} \sum_{i=1}^C P_{y=i}^{(k)} \nabla_{\bm{w}} \mathbb{E}_{\boldsymbol{x} \mid y=i}\left[\log f_i(\boldsymbol{x}, \bm{w}_{r,e}^{(k)})\right],
\end{align}
where \(P_{y=i}^{(k)}\) represents the data proportion of class \(i\) on clients \(k\). 


To derive the convergence bound of the loss function, we introduce a cell-centralized SGD algorithm, where each ES collects the datasets of its associated clients, performs local
model training, and then uploads the updates to the CS for aggregation. Let \(\bm{w}_{r}^{(c)}\) denote the global model of \(r\)-th round, and \(\bm{w}_{r,e+1}^{(c_l)}\) represent the \(e\)-th update in the \(r\)-th round under the centralized setting. Then cell-centralized SGD performs the following update

\begin{align} 
\label{equ21}
\bm{w}_{r,e+1}^{(c_l)} &= \bm{w}_{r,e}^{(c_l)} - \eta_{r,e}  \sum_{i=1}^C P_{y=i}^{(c_l)} \nabla_{\bm{w}} \mathbb{E}_{\boldsymbol{x} \mid y=i}\left[\log f_i(\boldsymbol{x}, \bm{w}_{r,e}^{(c_l)})\right],
\end{align}

where \(e = 0,1,\ldots,E-1\), and \(P_{y=i}^{(c_l)}\) represents the data distribution of participating clients in cell \(l\). After \(E\) times update, the aggregated model for next round \(r+1\) in centralized settings can be written as

\begin{align} 
\label{equ22}
\bm{w}_{r+1}^{(c)} = \sum_{i=1}^{L}\frac{\hat{N}_{r}^{(f_i)} \bm{w}_{r,E}^{(c_i)}}{\sum_{i=1}^{L} \hat{N}_{r}^{(f_i)}}.
\end{align}

Combining \textbf{Assumption 1} with \(\sum_{i=1}^{C} P_{y=i}\lambda_{x|y=i} = \lambda\), 
we obtain that \(\ell(\bm{w})\) is \(\lambda\)-smooth, i.e., \(\|\nabla \ell(\bm{w}) - \nabla \ell(\bm{v})\| \leq \lambda \|\bm{w} - \bm{v}\|,~\forall~\bm{v}, \bm{w}.\) Then we have

\begin{align} 
\label{equ23}
\ell(\bm{w}_{R}^{(f_l)}) - \ell(\bm{w^{*}}) &\leq \frac{\lambda}{2} \left\| \bm{w}_{R}^{(f_l)} - \bm{w}^{*} \right\|\notag\\
& = \frac{\lambda}{2} \left\| \bm{w}_{R}^{(f_l)} - \bm{w}_{R}^{(c)} + \bm{w}_{R}^{(c)} - \bm{w}^{*} \right\| \notag\\
& \leq \frac{\lambda}{2} \underbrace{\left\|\bm{w}_{R}^{(f_l)} - \bm{w}_{R}^{(c)}\right\|}_{A_1} 
+  \frac{\lambda}{2} \underbrace{\left\| \bm{w}_{R}^{(c)} - \bm{w}^{*} \right\|}_{A_2},\notag\\
\end{align} 

We next focus on bounding \(A_1\). Based on the definitions of \(\bm{w}_{R}^{(f_l)} \text{ and } \bm{w}_{R}^{(c)}\) in~\eqref{equ4} and~\eqref{equ22}, we have

\begin{align}
\label{equ24}
A_1 &=  \bigg\|\bm{w}_{R}^{(f_l)} - \bm{w}_{R}^{(c)}\bigg\| \notag\\
&=  \bigg\| 
\sum_{j=1}^{L}  
\frac{p_{R-1}^{(j,l)}\hat{N}_{R-1}^{(f_j)} \bm{\hat{w}}_{R-1,E}^{(f_j)}}{\sum_{j=1}^{L} p_{R-1}^{(j,l)}\hat{N}_{R-1}^{(f_j)}} 
- \sum_{j=1}^{L} 
\frac{\hat{N}_{R-1}^{(f_j)} \bm{w}_{R-1,E}^{(c_j)}}{\sum_{j=1}^{L} \hat{N}_{R-1}^{(f_j)}} \bigg\|\notag\\
&=\bigg\| 
\sum_{j=1}^{L} 
\frac{\hat{N}_{R-1}^{(f_j)}}{\sum_{j=1}^{L} \hat{N}_{R-1}^{(f_j)}} 
\bigg( 
\frac{p_{R-1}^{(j,l)} \sum_{j=1}^{L} \hat{N}_{R-1}^{(f_j)}}{\sum_{j=1}^{L} p_{R-1}^{(j,l)} \hat{N}_{R-1}^{(f_j)}} 
\bm{\hat{w}}_{R-1,E}^{(f_j)} \notag\\
&\quad - \bm{\hat{w}}_{R-1,E}^{(f_j)} + \bm{\hat{w}}_{R-1,E}^{(f_j)} 
- \bm{w}_{R-1,E}^{(c_j)} \bigg) \bigg\|  \notag\\
&=\bigg\|
\sum_{j=1}^{L} 
\frac{\hat{N}_{R-1}^{(f_j)}}{\sum_{j=1}^{L} \hat{N}_{R-1}^{(f_j)}} 
\bigg[ 
\bigg( 
\frac{p_{R-1}^{(j,l)} 
\sum_{j=1}^{L} 
\hat{N}_{R-1}^{(f_j)}}{\sum_{j=1}^{L} p_{R-1}^{(j,l)} 
\hat{N}_{R-1}^{(f_j)}} - 1
\bigg) 
\bm{\hat{w}}_{R-1,E}^{(f_j)} \notag\\
&\quad + \left( 
\bm{\hat{w}}_{R-1,E}^{(f_j)} 
- \bm{w}_{R-1,E}^{(c_j)} 
\right)
\bigg] 
\bigg\|  \notag\\
& \leq  \sum_{j=1}^{L} 
\bigg| 
\frac{p_{R-1}^{(j,l)}\hat{N}_{R-1}^{(f_j)}}{\sum_{j=1}^{L} p_{R-1}^{(j,l)}\hat{N}_{R-1}^{(f_j)}} - \frac{\hat{N}_{R-1}^{(f_j)}}{\sum_{j=1}^{L} \hat{N}_{R-1}^{(f_j)}} 
\bigg| 
\left\| \bm{\hat{w}}_{R-1,E}^{(f_j)} \right\|\notag\\
&\quad + 
\underbrace{\bigg\| 
\sum_{j=1}^{L} 
\frac{\hat{N}_{R-1}^{(f_j)}}{\sum_{j=1}^{L} \hat{N}_{R-1}^{(f_j)}} 
\left(\bm{\hat{w}}_{R-1,E}^{(f_j)} - \bm{w}_{R-1,E}^{(c_j)} \right)
\bigg\|}_{B_1}.
\end{align}

We substitute~\eqref{equ5},~\eqref{equ21} into \(B_1\), and we have

\begin{align} 
\label{equ25}
& \quad B_1 \notag\\
&= \bigg\| 
\sum_{j=1}^{L} 
\frac{\hat{N}_{R-1}^{(f_j)}}{\sum_{j=1}^{L} \hat{N}_{R-1}^{(f_j)}} 
\left(\bm{\hat{w}}_{R-1,E}^{(f_j)} - \bm{w}_{R-1,E}^{(c_j)} \right)
\bigg\|\notag\\
&\quad \cdot \sum_{i=1}^C P_{y=i}^{(k)} \nabla_{\bm{w}} \mathbb{E}_{\boldsymbol{x} \mid y=i}\left[\log f_i(\boldsymbol{x}, \bm{w}_{R-1,E-1}^{(k)})\right] \bigg) - \bm{w}_{R-1,E-1}^{(c_j)}\notag\\
&\quad  + \eta_{R-1,E-1}  \sum_{i=1}^C P_{y=i}^{(c_j)} \nabla_{\bm{w}} \mathbb{E}_{\boldsymbol{x} \mid y=i}\left[\log f_i(\boldsymbol{x}, \bm{w}_{R-1,E-1}^{(c_j)})\right]
\bigg] \bigg\|\notag\\
& \stackrel{1}= \bigg\| 
\sum_{j=1}^{L} 
\frac{\hat{N}_{R-1}^{(f_j)}}{\sum_{j=1}^{L} \hat{N}_{R-1}^{(f_j)}} 
\sum_{k=1}^{\hat{K}_{r}^{(f_j)}} 
\frac{n^{(k)}}{\hat{N}_{R-1}^{(f_j)}} 
\left[ \left( \bm{w}_{R-1,E-1}^{(k)} - \bm{w}_{R-1,E-1}^{(c_j)} \right) \right. \notag\\
& \quad + \left. \eta_{R-1,E-1} \sum_{i=1}^C P_{y=i}^{(k)} 
\left( \nabla_{\bm{w}} \mathbb{E}_{\boldsymbol{x} \mid y=i} \left[\log f_i(\boldsymbol{x}, \bm{w}_{R-1,E-1}^{(k)}) \right] \right. \right. \notag\\
& \quad \left. \left. - \nabla_{\bm{w}} \mathbb{E}_{\boldsymbol{x} \mid y=i} \left[\log f_i(\boldsymbol{x}, \bm{w}_{R-1,E-1}^{(c_j)}) \right] \right)
\right] 
\bigg\| \notag\\
& \stackrel{2}\leq 
\sum_{j=1}^{L} 
\frac{\hat{N}_{R-1}^{(f_j)}}{\sum_{j=1}^{L} \hat{N}_{R-1}^{(f_j)}} 
\sum_{k=1}^{\hat{K}_{r}^{(f_j)}} 
\frac{n^{(k)}}{\hat{N}_{R-1}^{(f_j)}} \cdot \notag\\
& \quad \left(
1 + \eta_{R-1,E-1} \sum_{i=1}^{C} P_{y=i}^{(k)} \lambda_{\bm{x} \mid y=i}
\right) \notag\\
& \quad \cdot \left\|
\bm{w}_{R-1,E-1}^{(k)} - \bm{w}_{R-1,E-1}^{(c_j)}
\right\|
\end{align} 
where equality 1 holds because for each class \(i \in \{1, 2, \ldots, C\}\), \( P_{y=1}^{(c_j)} = \sum_{k=1}^{\hat{K}_{r}^{(f_j)}} 
\frac{n^{(k)}}{\hat{N}_{R}^{(f_j)}} P_{y=1}^{(k)}\), i.e., the data distribution over all the clients in cell \(j\) is the same as the data distribution over the whole cell population. Inequality 2 holds because of the \textbf{Assumption 1}. Inspired by~\cite{a34}, we give the bounding of \(\left\| \bm{w}_{R-1,E-1}^{(k)} - \bm{w}_{R-1,E-1}^{(c_j)} \right\|\) for client \(k \in \hat{\mathcal{K}}_{r}^{(c_j)}\) in cell \(j\):

\begin{align} 
\label{equ26}
&\bigg\| \boldsymbol{w}_{R-1,E-1}^{(k)} - \boldsymbol{w}_{R-1,E-1}^{(c_j)} \bigg\| \notag\\
\stackrel{3}=& \bigg\| \left(\boldsymbol{w}_{R-1,E-2}^{(k)} -  \boldsymbol{w}_{R-1,E-2}^{(c_j)} \right)\notag\\
& + \bigg( \eta_{R-1,E-2} \sum_{i=1}^C P_{y=i}^{(k)} \nabla_{\boldsymbol{w}} \mathbb{E}_{\boldsymbol{x} \mid y=i} \left[ \log f_i(\boldsymbol{x}, \boldsymbol{w}_{R-1,E-2}^{(c_j)}) \right]  \notag\\
& - \eta_{R-1,E-2} \sum_{i=1}^C P_{y=i}^{(k)} \nabla_{\boldsymbol{w}} \mathbb{E}_{\boldsymbol{x} \mid y=i} \left[ \log f_i(\boldsymbol{x}, \boldsymbol{w}_{R-1,E-2}^{(k)}) \right] \bigg) \bigg\| \notag\\
\leq& \bigg\| \boldsymbol{w}_{R-1,E-2}^{(k)} - \boldsymbol{w}_{R-1,E-2}^{(c_j)} \bigg\| \notag\\ 
& + \eta_{R-1,E-2} \bigg\| \sum_{i=1}^C P_{y=i}^{(k)} \nabla_{\boldsymbol{w}} \mathbb{E}_{\boldsymbol{x} \mid y=i} \left[ \log f_i(\boldsymbol{x}, \boldsymbol{w}_{R-1,E-2}^{(k)}) \right] \notag\\
& - \sum_{i=1}^C P_{y=i}^{(k)} \nabla_{\boldsymbol{w}} \mathbb{E}_{\boldsymbol{x} \mid y=i} \left[ \log f_i(\boldsymbol{x}, \boldsymbol{w}_{R-1,E-2}^{(c_j)}) \right] \notag\\
& + \sum_{i=1}^C P_{y=i}^{(k)} \nabla_{\boldsymbol{w}} \mathbb{E}_{\boldsymbol{x} \mid y=i} \left[ \log f_i(\boldsymbol{x}, \boldsymbol{w}_{R-1,E-2}^{(c_j)}) \right] \notag\\
& - \sum_{i=1}^C P_{y=i}^{(c_j)} \nabla_{\boldsymbol{w}} \mathbb{E}_{\boldsymbol{x} \mid y=i} \left[ \log f_i(\boldsymbol{x}, \boldsymbol{w}_{R-1,E-2}^{(c_j)}) \right] \bigg\| \notag\\
\leq & \bigg\| \boldsymbol{w}_{R-1,E-2}^{(k)} - \boldsymbol{w}_{R-1,E-2}^{(c_j)} \bigg\| \notag\\
& + \eta_{R-1,E-2} \bigg[ \bigg\| \sum_{i=1}^C P_{y=i}^{(k)} \bigg( \nabla_{\boldsymbol{w}} \mathbb{E}_{\boldsymbol{x} \mid y=i} \left[ \log f_i(\boldsymbol{x}, \boldsymbol{w}_{R-1,E-2}^{(k)}) \right] \notag\\
& - \nabla_{\boldsymbol{w}} \mathbb{E}_{\boldsymbol{x} \mid y=i} \left[ \log f_i(\boldsymbol{x}, \boldsymbol{w}_{R-1,E-2}^{(c_j)}) \right] \bigg) \bigg\| + \notag\\
& \bigg\| \sum_{i=1}^C \left( P_{y=i}^{(k)} - P_{y=i}^{(c_j)} \right)  \nabla_{\boldsymbol{w}} \mathbb{E}_{\boldsymbol{x} \mid y=i} \left[ \log f_i(\boldsymbol{x}, \boldsymbol{w}_{R-1,E-2}^{(c_j)}) \right] \bigg\| \bigg]\notag\\
\stackrel{4}\leq& \bigg(1 + \eta_{R-1,E-2} \sum_{i=1}^C P_{y=i}^{(k)} \lambda_{x\mid y=i}\bigg) \bigg\| \boldsymbol{w}_{R-1,E-2}^{(k)} - \boldsymbol{w}_{R-1,E-2}^{(c_j)} \bigg\| \notag\\
& + \eta_{R-1,E-2} \sum_{i=1}^C \bigg| P_{y=i}^{(k)} - P_{y=i}^{(c_j)} \bigg|  g_{\max} (\boldsymbol{w}_{R-1,E-2}^{(c_j)}) \notag\\
\stackrel{5}=& a_{R-1,E-2}^{(k)} \bigg\| \boldsymbol{w}_{R-1,E-2}^{(k)} - \boldsymbol{w}_{R-1,E-2}^{(c_j)} \bigg\| \notag\\
& +  \eta_{R-1,E-2} \sum_{i=1}^C \bigg| P_{y=i}^{(k)} - P_{y=i}^{(c_j)} \bigg| g_{\max} (\boldsymbol{w}_{R-1,E-2}^{(c_j)}) 
\end{align}

where equality 3 holds because \eqref{equ20} and \eqref{equ21}. Inequality 4 holds because the \textbf{Assumption 1} and we denote \(g_{\max }\left(\boldsymbol{w}_{r, q}^{(c_j)}\right)=\max _{i=1}^C\left\|\nabla_{\boldsymbol{w}} \mathbb{E}_{\boldsymbol{x} \mid y=i} \log f_i\left(\boldsymbol{x}, \boldsymbol{w}_{r, q}^{(c_j)}\right)\right\|\) for \(r \in \{1,2,\ldots,R\} \text{ and } q \in \{0,1,\ldots,E\}\). equality 5 holds because we denote \(a_{r,q}^{(k)}=1+\eta_{r,q} \sum_{i=1}^C P_{y=i}^{(k)} \lambda_{\boldsymbol{x} \mid y=i}\) for \(r \in \{1,2,\ldots,R\} \text{ and } q \in \{0,1,\ldots,E\}\). From \eqref{equ26} we deduce the relation between \( \left\| \boldsymbol{w}_{R-1,E-1}^{(k)} - \boldsymbol{w}_{R-1,E-1}^{(c_j)} \right\| \text{ and } \left\| \boldsymbol{w}_{R-1,E-2}^{(k)} - \boldsymbol{w}_{R-1,E-2}^{(c_j)} \right\| \). 

By repeatedly using this relation, we can recursively obtain:

\begin{align}
\label{equ27}
&\bigg\| \boldsymbol{w}_{R-1,E-1}^{(k)} - \boldsymbol{w}_{R-1,E-1}^{(c_j)} \bigg\| \notag\\
\leq& a_{R-1,E-2}^{(k)} \bigg\| \boldsymbol{w}_{R-1,E-2}^{(k)} - \boldsymbol{w}_{R-1,E-2}^{(c_j)} \bigg\| \notag\\
& + \sum_{i=1}^C \bigg| P_{y=i}^{(k)} - P_{y=i}^{(c_j)} \bigg| \eta_{R-1,E-2} g_{\max} (\boldsymbol{w}_{R-1,E-2}^{(c_j)}) \notag\\
\leq& \left( \prod_{q=0}^{E-2} a_{R-1,q}^{(k)}\right)
\left\|\boldsymbol{w}_{R-1, 0}^{(k)}-\boldsymbol{w}_{R-1, 0}^{(c_j)}\right\| 
+ \sum_{i=1}^C\left|P_{y=i}^{(k)}-P_{y=i}^{(c_j)}\right|\notag\\
&\quad \cdot \sum_{q=0}^{E-2} \eta_{R-1,q} 
\left(\prod_{t=q+1}^{E-2} a_{R-1,t}^{(k)}\right) 
g_{\max}\!\left(\boldsymbol{w}_{R-1,q}^{(c_j)}\right) \notag\\
\leq& \left( \prod_{q=0}^{E-2} a_{R-1,q}^{(k)}\right)
\left\|\boldsymbol{w}_{R-1}^{(f_j)}-\boldsymbol{w}_{R-1}^{(c)}\right\|
+ \sum_{i=1}^C\left|P_{y=i}^{(k)}-P_{y=i}^{(c_j)}\right|\notag\\
&\quad \cdot \sum_{q=0}^{E-2} \eta_{R-1,q} 
\left(\prod_{t=q+1}^{E-2} a_{R-1,t}^{(k)}\right) 
g_{\max}\!\left(\boldsymbol{w}_{R-1,q}^{(c_j)}\right),
\end{align}
where we assume \(\prod_{a}^{b} a_{r,e}^{(k)} = 1,~\forall b < a,~r~\text{and}~e\). Institute \eqref{equ27} into \eqref{equ25}, we have

\begin{align} 
\label{equ28}
B_1
&\leq \sum_{j=1}^{L} 
\frac{\hat{N}_{R-1}^{(f_j)}}{\sum_{j=1}^{L} \hat{N}_{R-1}^{(f_j)}} 
\sum_{k=1}^{\hat{K}_{r}^{(f_j)}} 
\frac{n^{(k)}}{\hat{N}_{R-1}^{(f_j)}} 
\Bigg[ \left( \prod_{e=0}^{E-1} a_{R-1,e}^{(k)}\right)\notag\\ 
&\quad \cdot \left\|\boldsymbol{w}_{R-1}^{(f_j)}-\boldsymbol{w}_{R-1}^{(c)}\right\| 
+ \sum_{i=1}^C\left|P_{y=i}^{(k)}-P_{y=i}^{(c_j)}\right|\notag\\
&\quad \cdot \sum_{d=0}^{E-2}  \eta_{R-1,d}  
\left(\prod_{e=d+1}^{E-1} a_{R-1,e}^{(k)}\right) 
g_{\max }\left(\boldsymbol{w}_{R-1,d}^{(c_j)}\right) \Bigg]\notag\\
&= \sum_{j=1}^{L} 
\bigg( D_{R-1}^{(j)}  
\left\|\boldsymbol{w}_{R-1}^{(f_j)}-\boldsymbol{w}_{R-1}^{(c)}\right\| 
+ G_{R-1}^{(j)} \bigg), \notag\\
\end{align}
where \(D_{R-1}^{(j)} 
= \frac{\hat{N}_{R-1}^{(f_j)}}{\sum_{j=1}^{L} \hat{N}_{R-1}^{(f_j)}}  
\sum_{k=1}^{\hat{K}_{r}^{(f_j)}} 
\frac{n^{(k)}}{\hat{N}_{R-1}^{(f_j)}}  
\left( \prod_{e=0}^{E-1} a_{R-1,e}^{(k)}\right)\),
and \(G_{R-1}^{(j)} 
= \frac{\hat{N}_{R-1}^{(f_j)}}{\sum_{j=1}^{L} \hat{N}_{R-1}^{(f_j)}}  
\sum_{k=1}^{\hat{K}_{r}^{(f_j)}} 
\frac{n^{(k)}}{\hat{N}_{R-1}^{(f_j)}} 
\left(\sum_{i=1}^C\bigl|P_{y=i}^{(k)}-P_{y=i}^{(c_j)}\bigr|\right) \\
\sum_{e=0}^{E-2} \eta_{R-1,d} 
\left(\prod_{d=e+1}^{E-1} a_{R-1,d}^{(k)}\right) 
g_{\max }\!\left(\boldsymbol{w}_{R-1,e}^{(c_j)}\right) \). Plug~\eqref{equ28} into~\eqref{equ24}, we have

\begin{align}
\label{equ29}
A_1 &=  \bigg\|\bm{w}_{R}^{(f_l)} - \bm{w}_{R}^{(c)}\bigg\| \notag\\ 
& \leq \sum_{j=1}^{L} 
\bigg| 
\frac{p_{R-1}^{(j,l)}\hat{N}_{R-1}^{(f_j)}}{\sum_{j=1}^{L} p_{R-1}^{(j,l)}\hat{N}_{R-1}^{(f_j)}} - \frac{\hat{N}_{R-1}^{(f_j)}}{\sum_{j=1}^{L} \hat{N}_{R-1}^{(f_j)}} 
\bigg| 
\left\| \bm{\hat{w}}_{R-1}^{(f_j)} \right\|\notag\\
& \quad + \sum_{j=1}^{L} 
\bigg( D_{R-1}^{(j)}  \left\|\boldsymbol{w}_{R-1}^{(f_j)}-\boldsymbol{w}_{R-1}^{(c)}\right\| + G_{R-1}^{(j)} \bigg) \notag\\
& \stackrel{6}= \sum_{j=1}^{L} 
D_{R-1}^{(j)}  \left\|\boldsymbol{w}_{R-1}^{(f_j)}-\boldsymbol{w}_{R-1}^{(c)}\right\| + F_{R-1}^{(l)} + G_{R-1} \notag\\
\end{align}
where equality 6 holds because we denote 

\begin{align}
\label{equ30}
F_{R-1}^{(l)} = \sum_{j=1}^{L}\left| \frac{p_{R-1}^{(j,l)}\hat{N}_{R-1}^{(f_j)}}{\sum_{j=1}^{L} p_{R-1}^{(j,l)}\hat{N}_{R-1}^{(f_j)}} - \frac{\hat{N}_{R-1}^{(f_j)}}{\sum_{j=1}^{L} \hat{N}_{R-1}^{(f_j)}} 
\right| \left\| \bm{\hat{w}}_{R-1,E}^{(f_j)} \right\|,
\end{align}
and \(G_{R-1} = \sum_{j=1}^{L} G_{R-1}^{(j)}\). 

In the following, we first derive a theoretical upper bound for \(G_{R-1}\). We denote
\begin{align}
\label{equ31}
G_{R-1}^{(j)}
&= \sum_{e=0}^{E-2} \eta_{R-1,e}\,\beta_{R-1,e}^{(j)},
\end{align}
and
\begin{align}
\label{equ32}
\beta_{R-1,e}^{(j)}=& \frac{\hat{N}_{R-1}^{(f_j)}}{\sum_{j=1}^{L} \hat{N}_{R-1}^{(f_j)}}  
\sum_{k=1}^{\hat{K}_{r}^{(f_j)}} 
\frac{n^{(k)}}{\hat{N}_{R-1}^{(f_j)}} 
\left(\sum_{i=1}^C\bigl|P_{y=i}^{(k)}-P_{y=i}^{(c_j)}\bigr|\right) \notag\\
&\cdot\left(\prod_{d=e+1}^{E-1} a_{R-1,d}^{(k)}\right) 
g_{\max }\!\left(\boldsymbol{w}_{R-1,e}^{(c_j)}\right). 
\end{align}

We assume that there exist constants \(g_{\max}\) and \(\beta_{r}^{(j)}\) such that \(g_{\max}\big(\bm{w}_{r,e}^{(c_j)}\big) \le g_{\max}\) and \(\beta_{r,e}^{(j)} \le \beta_{r}^{(j)}\), \(\forall r,~e,~j\). Moreover, we denote \(\sum_{j=1}^{L} \beta_{r}^{(j)} = \beta_{r}\). Since \(a_{r,q}^{(k)}=1+\eta_{r,q} \sum_{i=1}^C P_{y=i}^{(k)} \lambda_{\boldsymbol{x} \mid y=i} >1,~\forall r,~q\), based on \eqref{equ25}, we further assume \(\beta_{\text{min}} \leq \beta_{r} \leq \beta_{\text{max}}\). Then we have

\begin{align} 
\label{equ33}
G_{R-1}= \sum_{e=0}^{E-2}\eta_{R-1,e}\sum_{j=1}^{L}\beta_{R-1,e}^{(j)}  
\leq\sum_{e=0}^{E-2} \eta_{R-1,e} \beta_{R-1} \leq   \frac{\beta_{R-1}}{R},
\end{align}

Finally, we will give the upper bound of \(A2\). \(A2\) is the model divergence between the cell-centralized SGD algorithm and \(\bm{w}^{*}\). If we see each cell as a client, the cell-centralized SGD is the standard FedAvg. Under the assumption that \(\bm{w}^{*}\) is the final model of the global centralized SGD, where data in all clients are collected in the CS to train a global model. Under the assumptions of \textbf{Theorem 1}, \cite{a34} gives the bound of \(A_2\) as 
\begin{align}
\label{equ34}
A_2 &= \left\| \bm{w}_{R}^{(c)} - \bm{w}^{*} \right\| \notag\\
&\leq \frac{\sum_{j=1}^{3} \hat{N}_{R-1}^{(f_j)} H^{(j)}\left( \sum_{i=1}^{C} \left\| P_{y=i}^{(c_j)} - P_{y=i}^{(c)} \right\|\right) }{NR(E-1)}. 
\end{align}
where \(N\) is the total data volume of all clients and \(H^{(j)}= \sum_{e=0}^{E-1} \left(\prod_{d=e+1}^{E-1}a_{R-1,d}^{(c_j)} \right) g_{\text{max}}( \bm{w}_{R-1,e}^{(c)})\).

Substituting~\eqref{equ29},~\eqref{equ33} and \eqref{equ34} into \eqref{equ23}, we have

\begin{align} 
\label{equ35}
& \ell(\bm{w}_{R}^{(f_l)}) - \ell(\bm{w^{*}}) \notag\\
\leq & \frac{\lambda}{2} \Bigg[ \sum_{j=1}^{L} 
D_{R-1}^{(j)}  \left\|\boldsymbol{w}_{R-1}^{(f_j)}-\boldsymbol{w}_{R-1}^{(c)}\right\| + \underbrace{\frac{\beta_{R-1}}{R}}_{\epsilon^{\text{inter}}} \notag\\
&+ \underbrace{\frac{\sum_{j=1}^{3} \hat{N}_{R-1}^{(f_j)} H^{(j)}\left( \sum_{i=1}^{C} \left\| P_{y=i}^{(c_j)} - P_{y=i}^{(c)} \right\|\right) }{NR(E-1)} + F_{R-1}^{(l)}}_{\epsilon^{\text{inter}}} \Bigg]
\end{align} 
which completes the proof.
\end{proof}

\bibliographystyle{IEEEtran}  
\bibliography{references}      


\end{document}